\documentclass[lnicst]{svmultln}
\usepackage{makeidx}  
\usepackage{amsmath}
\usepackage{appendix}
\usepackage{url}

\begin{document}
\mainmatter              
\title{Stable and Efficient Structures for the Content Production and
Consumption in Information Communities}
\titlerunning{Stable and Efficient Structures}
%
\author{Larry Yueli Zhang\inst{1} \and Peter Marbach\inst{2}}
\authorrunning{Zhang et al.}   
%
\tocauthor{Larry Yueli Zhang, Peter Marbach}
\institute{
	University of Toronto\\
\email{ylzhang@cs.toronto.edu}
\and
University of Toronto\\
\email{marbach@cs.toronto.edu}
}

\maketitle              

\begin{abstract}        
	Real-world information communities exhibit inherent structures that
	characterize a system that is stable and efficient for content
	production and consumption. In this paper, we study such structures
	through mathematical modelling and analysis. We formulate a generic
	model of a community in which each member decides how they allocate
	their time between content production and consumption with the objective
	of maximizing their individual reward. We define the community system as
	``stable and efficient'' when a Nash equilibrium is reached while the
	social welfare of the community is maximized. We investigate the
	conditions for forming a stable and efficient community under two
	variations of the model representing different internal relational
	structures of the community. Our analysis results show that the
	structure with ``a small core of celebrity producers'' is the optimally
	stable and efficient for a community. These analysis results provide
	possible explanations to the sociological observations such as ``the Law
	of the Few'' and also provide insights into how to effectively build and
	maintain the structure of information communities.
\end{abstract}
\section{Introduction}

Communities are an important structure that widely exists in real-world
online and offline social networks. A common type of community is the
\emph{information community} in which the members of the community produce
content and consume the content produced by other members, with the most
popular example being Reddit~\cite{reddit} where each ``subreddit'' is
essentially an information community with a specific topic of interest.
Real-world communities often exhibit inherent structures such as the high
density of interactions within the community and the existence of a core set
of active members who would contribute the majority of the content in the
community (``the Law of the Few'')~\cite{gladwell, zhang}. There have been a
large body of research work on community detection algorithms based on such
structures. However, there is still a lack of the formal understanding of
why these structures would consistently and naturally emerge during the
formation process of real-world communities. Understanding the formation
process of these natural structures is important as it provides us a
microscopic view of the working mechanisms of communities and would enable
us to utilize communities more efficiently.

Our overall hypothesis is the following: real-world social network
structures have been going through an evolutionary process, and as a result
of that only the optimal structure (in terms of stability and efficiency)
can survive, sustain therefore exist widely in real-world social networks.
In other words, if we observe a widely existing structure in real-world
social networks, then this structure must be optimal in the sense that it
has stable user behaviours and it is efficient for the purpose of the
network.

In the case of information communities, each member in the community is an
agent who can choose to spend certain portions of their time in producing
content items or in consuming content produced by other members. In order
for the community to be stable, all members' time allocation strategies
should collectively form a Nash equilibrium, i.e., each member would get
penalized by deviating from the equilibrium strategy. A member in the
community can be rewarded by either production or consumption. For
consumption, the rewarded is from the consumed content itself; for
production, a member is rewarded when the content she produces is consumed
by other members of the community (the reputation effect). A community
structure is called ``efficient'' when it can provide its members the
highest possible amount of reward. If we use a mathematical model to
formulate the above behaviours and efficiency measures, we will then be able
to formally analyze the condition under which the community structure is
optimally stable and efficient, therefore obtain a mathematical description
of the ``surviving and sustaining'' community structure. The validity of the
model would be verified if the result of the analysis happens to agree with
the widely-existing structures observed in real-world communities.  Compared
to the empirical observations, the formal analytical results would provide
us more refined understanding of the microscopic working mechanisms of the
real-world communities.

In this paper, we formulate a model that captures the production and
consumption behaviours inside an information community. Our analysis results
show that the structure with a small set of ``celebrity producers'' is the
optimally stable and efficient structure. These analysis results provide
possible explanations to the sociological observations such as ``the Law of
the Few'' and also provide insights into how to effectively build and
maintain the structure of information communities.

\section{Related Works}

Social network analysis has been one of the fastest growing research fields
in the 21st century. We refer readers to Scott et al.~\cite{scott} for a
comprehensive coverage of the development of the subject, rather than
listing the large collection of references in this paper. Experimental works
observed interesting properties of real-world complex networks such as the
power-law degree distribution, the small-world phenomena and the community
structure. These observations lead to modelling works that tried to explain
why the observed properties would emerge, such models include the
preferential attachment models, the copying model and the forest fire model.
However, most of these works were studying macroscopic structural properties
rather than looking into the internal microscopic structures of the network.

The community structure has been an interesting topic for researcher in the
field of social network analysis. A large body of work has been devoted to
modelling and detecting community structures in large scale social networks
(e.g., \cite{fortunato, massoulie, newman, leskovec2010, kumar}). The
networks are often represented by graphs in which the vertices represent
underlying social entities and the edges represent some sort of social tie
or interaction between pairs of vertices. Our model differs in the sense
that it also considers the user behaviours on top of the network
connections.

In~\cite{bosagh}, the efficiency of a network in terms of information
diffusion is studied, a mathematical analysis is perform to investigate the
optimal network structure to achieve the best efficiency for information
diffusion (high precision, high recall and low diameter), and the result
shows that a Kronecker-graph~\cite{leskovec2010} would satisfy such
conditions. The approach taken in \cite{bosagh} is similar to the approach
we take in this paper except that we are more focussed on the community
related aspects.  The work in~\cite{galeotti} used a game-theoretic model to
study the emergence of the ``Law of the Few'' but it is also in the context
of information diffusion rather than about communities.  The work
in~\cite{marbach} is the closest to the interest of this paper.
In~\cite{marbach}, a game theoretic model is formulated to analyze the
community structures in terms of content production and consumption. Each
member's strategy involves choosing a particular interest to produce or
consume content on. The result shows that in the Nash equilibrium of the
model the members' choices form community structures. The difference of our
work from \cite{marbach} is that we focus on the internal structure of a
single community rather than on the scale of multiple communities, and
besides the Nash equilibrium, we also take the social welfare into
consideration.

\section{Model}

We will first describe the general configuration of the model and the
payoff/reward functions, then in Section~\ref{subsec-model2} and
Section~\ref{subsec-model3}, we introduce two variations of modelling the
internal relations between the community members. Both models will be
analyzed and the results will be compared in Section~\ref{sec-analysis}.

\subsection{General Configuration}

We have a single community with $n$ members indexed by $1\le i\le n$. Each
member is capable of both producing and consuming content items. The
produced content items could be chosen by all members or a subset of the
members of the community for consumption. Each member has a limited total
amount of time which could be allocated to either production or
consumption, and each member make a decision about how much of their time to
allocate to production and consumption. A member is rewarded if their
products are consumed by members of the community (the production reward),
or if the member consumes an item that is produced by a member of the
community (the consumption reward). Each member's objective is to maximize
their total individual reward from both production and consumption.

The time slot: In our model, we investigate everything that happens within a
\emph{unit time}. The assumption is that the long term behaviour of a member
is the repetition of their behaviour within a unit time.

Rates of content production and consumption: We define $N_p\ge 0$ to be the
number of content items that a member can produce if they were to spend
100\% of their unit time on production; and we let $N_c\ge 0$ be the number
of items that a member can consume within a unit time slot if they were to
spend 100\% of their time on consumption. We assume that all members
share the same values of $N_p$ and $N_c$, and we assume the following
inequality:
\begin{equation}\label{assum-nc}
	0\le \frac{N_c}{nN_p} \le 1
\end{equation}
This assumption is reasonable because $nN_p$ is the largest possible number
of content items that can be produced, an $N_c$ value that is larger than
$nN_p$ would be unrealistic.

A member's time allocation strategy: Let $\alpha_i\; (0\le \alpha\le 1)$ be
the portion of the unit time that member $i$ allocates to production
(therefore $1-\alpha$ is allocated to consumption). Each member chooses
their own $\alpha_i$, we will investigate if a set of choices of $\alpha_i$
would lead to a Nash equilibrium. Within the unit time, a member can consume
at most $(1-\alpha_i)\cdot N_c$ items. If the number of available items is
less than or equal to this number, then each member would consume all the
available items without any choice; if the total number of available items is
greater than this number, then the member would choose a subset (of size
$(1-\alpha_i)\cdot N_c$) of the available items to consume, uniformly at
random.

The \emph{production reward} models the ``reputation effect'' in social
networks, i.e., having content products consumed by other people is
rewarding for the producer of the content. The reward for each item that a
member produces is proportional to the number of members who consume the
item, with a constant factor $r_p$, i.e., if an item is consumed by $m$
members, then the reward for this item is $r_p\cdot m$. The total production
reward for a member is the sum of the rewards of all items that the member
produces. The constant factor $r_p$ is the same for all members.  The
\emph{consumption reward} of a given item is a constant $r_c$. The total
consumption reward of a member is $r_c$ multiplied by the number of items
consumed by the member.  The total individual reward of a member in the
community is the sum of their production reward and consumption reward. The
sum of the total individual rewards of all members in the community is the
\emph{social welfare}. While each member tries to maximize their own
individual reward, the overall efficiency of the community is measured by
its social welfare.

The following two subsections will define two variations of the internal
relational structure of the community.

\subsection{The Celebrity-Follower Community Structure}\label{subsec-model2}

Under the celebrity-follower relational structure, a subset of the members
of the community are ``celebrities'' that are followed by everyone in the
community, i.e., the content items produced by a celebrity member can be
seen by all members of the community. A non-celebrity member has zero
followers, i.e., an item produced by a non-celebrity member cannot be seen
or consumed by any member.

Let $\eta$ be the portion of celebrity members, i.e., the number of
celebrity members is $\eta n$. When $\eta = 1$, all members are connected
via a complete graph. When $\eta$ is small, we
have a small core of celebrities that would be responsible for producing all
content items in the community. If visualized as a directed graph, the
structure would have $\eta\cdot n^2$ edges in total. Note that we are
assuming a member can be a follower of themselves so the graph can have
self-pointing edges. This would lead to cleaner analysis results.

Note that we are not making any assumptions about how large the value of
$\eta$ is, and it is interesting to see whether the efficiency of the
community system can be different with $\eta$'s value being in different
ranges. In real-world communities, we often observe patterns that are
similar to the celebrity-follower structure, i.e., a small subset of ``elite
contributors'' would produce most of the content items that are consumed by
all members of the community, and a community typically has a significant
portion of ``lurkers''. We will be able to provide a theoretical explanation
to this real-world phenomenon.

\subsection{The Uniform Community Structure}\label{subsec-model3}

In contrast to the celebrity-follower structure where the members play
unequal roles in the community, the uniform relational structure has all
members with the equal role, i.e., every member has the same number of
followers and follows the same number of other members. In terms of a graph,
it is a regular graph where every vertex has the same in-degrees and
out-degrees.

To make this structure comparable with the celebrity-follower structure, we
let it have the same number of edges as the celebrity-follower graph. The
celebrity-follower graph discussed in the previous section has $\eta\cdot
n^2$ edges, therefore, in the uniform graph, we let each vertex have
in-degree $\eta\cdot n$ as well as out-degree $\eta\cdot n$.

\subsection{Summary of the model}

Overall, our model is a game-theoretic model where each agent (member of the
community) chooses a strategy ($\alpha_i$) with the objective of optimizing
their individual reward. The efficiency of the whole community is measured
by the social welfare (total reward of all members). The stability of the
community is indicated by whether the strategies of all members collectively
form a Nash equilibrium.

\section{Analysis}\label{sec-analysis}

Our hypothesis is that, in order to exist and sustain in the real world, a
social structure must be stable and efficient. For an information community,
this means that the members' strategies form a Nash equilibrium while the
social welfare of the community is maximized. Therefore, our analysis will
take the following approach: we first derive the set of members' strategies
that would maximize the social welfare of the community, then we investigate
the condition for this set of strategies to form a Nash equilibrium.

In Section~\ref{subsec-analysis1} and Section~\ref{subsec-analysis2}, we
perform the analyses for the celebrity-follower and uniform structures,
respectively, then we will compare and discuss the analysis results.

\subsection{Analysis of the Celebrity-Follower Structure}\label{subsec-analysis1}

The following theorem summarizes the analysis results for communities with
the celebrity-follower structure.
\begin{theorem}\label{thm5}
	For a community with the celebrity-follower structure where there are
	$\eta\cdot n$ celebrity members, the maximum social welfare and the Nash
	equilibrium are described in the following cases.

	Case 1: if $\eta < \min(\frac{N_c}{nN_p}, 1 - \frac{N_c}{nN_p}, 1 -
	\frac{N_cr_c}{nN_pr_p})$, then the
	maximum social welfare is reached when a member $i$ of the community
	takes the following strategy:
	\begin{equation}
		\alpha_i =
		\begin{cases}
			1 \quad &\text{if member }i\text{ is a celebrity}\\
			0 \quad &\text{otherwise}
		\end{cases}
	\end{equation}
	The maximum social welfare $G_{\max}$ is the following:
	\begin{equation}\label{gmax-celeb}
		G_{max} = \eta(1-\eta)n^2 N_p (r_p + r_c)
	\end{equation}
	This set of strategies \textbf{always} form a Nash equilibrium under
	this case.

	Case 2: if $\frac{1}{2} < \frac{N_c}{nN_p} \le 1$ and $1-
	\frac{N_c}{nN_p}\le \eta \le \frac{N_c}{nN_p}$, then the maximum social
	welfare is reached when a member $i$ of the community follows the
	following strategy.
	\begin{equation}
		\alpha_i =
		\begin{cases}
			\frac{N_c}{N_c + \eta n N_p} \quad &\text{if member }i\text{ is
			a celebrity}\\
			0 \quad &\text{otherwise}
		\end{cases}
	\end{equation}
	The maximum social welfare $G_{\max}$ is the following:
	\begin{equation}
		G_{\max} = \frac{\eta n^2 N_cN_p(r_p + r_c)}{N_c + \eta n N_p}
	\end{equation}
	However, this set of strategies \textbf{never} forms a Nash equilibrium
	under this case.

	Case 3: if $0\le \frac{N_c}{nN_p}\le \frac{1}{2}$ and
	$\frac{N_c}{nN_p}\le \eta \le 1 - \frac{N_c}{nN_p}$, then the maximum
	social welfare is reached when a member $i$ of the community follows the
	following strategy.
	\begin{equation}
		\alpha_i =
		\begin{cases}
			\frac{N_c}{\eta n N_p} \quad &\text{if member }i\text{ is a celebrity}\\
			0 \quad &\text{otherwise}
		\end{cases}
	\end{equation}
	The maximum social welfare under this strategy is
	\begin{equation}
		G_{\max} = N_c\left(n - \frac{N_c}{N_p}\right)(r_p + r_c)
	\end{equation}
	This set of strategies \textbf{never} forms a Nash equilibrium under this
	case.

	Case 4: $\eta >  \max(\frac{N_c}{nN_p}, 1 - \frac{N_c}{nN_p})$, the
	social welfare is maximized when a member $i$ of the community
	follows the following strategy.
	\begin{equation}
		\alpha_i =
		\begin{cases}
			\frac{N_c}{N_c + \eta n N_p} \quad &\text{if member }i\text{ is
			a celebrity}\\
			0 \quad &\text{otherwise}
		\end{cases}
	\end{equation}
	The maximum social welfare under this strategy is
	\begin{equation}
		G_{\max} = \frac{\eta n^2 N_cN_p(r_p + r_c)}{N_c + \eta n N_p}
	\end{equation}
	This set of strategies \textbf{never} forms a Nash equilibrium under this
	case.
\end{theorem}

The detailed proof of Theorem~\ref{thm5} can be found in the appendix. This
theorem shows that Case 1 is the only case where the members' strategies
reach a Nash equilibrium while the social welfare is maximized. In other
words, in order for the community to be optimally stable and efficient, the
portion of celebrity members must be small enough, i.e., $\eta <
\min(\frac{N_c}{nN_p}, 1 - \frac{N_c}{nN_p}, 1 - \frac{N_cr_c}{nN_pr_p})$.

\subsection{Analysis of the Uniform Structure}\label{subsec-analysis2}

The following theorem summarizes the analysis results for communities with
the celebrity-follower structure.
\begin{theorem}\label{thm4}
	For a community with the uniform structure where each member has $\eta
	n$ followers and follows $\eta n$ members, the maximum social welfare
	is reached when the following set of strategies is applied.
	\begin{equation}
		\alpha_i = \frac{N_c}{N_c + \eta n N_p}\quad \forall 1\le i \le n
	\end{equation}
	The maximum social welfare $G_{\max}$ is the following:
	\begin{equation}\label{thm1globmax}
		G_{\max} = \frac{\eta n^2N_pN_c(r_c+r_p)}{N_c + \eta nN_p}
	\end{equation}
	The above set of strategies forms a Nash equilibrium if and only if the
	following condition is true.
	\begin{equation}
		\eta \le \left(\frac{N_cr_c}{nN_pr_p} + \frac{1}{n}\right)\label{pf3c1cond0}
	\end{equation}
\end{theorem}

The detailed proof of Theorem~\ref{thm4} can be found in the appendix. This
result shows that, assuming the uniform community structure, there exist a
simple set of strategies that is stable while the social welfare is
maximized. What we are interested in is how the optimal efficiency of the
uniform structure compares with that of a community with the
celebrity-follower structure. The following theorem provides us a formal
result.

\begin{theorem}\label{thm6}
	Let $G_{\max\text{-celebrity}}$ be the maximum social welfare with a
	Nash equilibrium for the celebrity-follower community structure
	(Eq~(\ref{gmax-celeb})) and $G_{\max\text{-uniform}}$ be the maximum
	social welfare with a Nash equilibrium for the uniform community
	structure (Eq~(\ref{thm1globmax})). The following is always true:
	\begin{equation}
		G_{\max\text{-celebrity}} \ge G_{\max\text{-uniform}}
	\end{equation}
\end{theorem}

The detailed proof of Theorem~\ref{thm6} can be found in the appendix. This
theorem provides a simple and clear result: given being in its optimally
stable and efficient state, a community with the celebrity-follower
structure always has a better optimal social welfare than a community with
the uniform structure.

\section{Discussions}

The combination of the analysis results in Section~\ref{subsec-analysis1}
and \ref{subsec-analysis2} provide us two different angles of explaining the
common ``law-of-the-few'' structural patterns that widely exist in real-life
information communities. A given community structure, in order to exist and
sustain, must be both stable and efficient, meaning that the community can
stably stay at the state with the maximum social welfare. Theorem~\ref{thm5}
tells us that the community can only be stable and efficient if there is a
small enough ``core'' of celebrity members who will actively contribute all
the content to be consumed by all members of the community, while the
majority of the community members would simply consume the content produced
by the core members. A community structure that does not satisfy this
condition would not be stable therefore would not commonly exist in reality.

Moreover, among the different possible structure that are both stable and
efficient, some structures are more efficient than others.
Theorem~\ref{thm6} shows that the small-core celebrity-follower structure is
not only stable and efficient, but also it is more efficient than other
stable structures such as the uniform structure.

With the above two factors taken into account, the celebrity-follower
structure with a small set of celebrities becomes the winner, therefore
becomes the commonly existing structure in real-world information
communities.

In the equilibrium state, the strategies of the celebrity and non-celebrity
members are clearly differentiated: the celebrity members should dedicate
all of their time in production whereas the non-celebrity members should
spend all of their time on consumption. These specialized producing and
consuming behaviours also coincide with real-world observations: in a web
service such as Reddit, the visitors of a typical subreddit would often
separate into two different roles, i.e., the ``active contributors'' who
frequently post content in the subreddit and the ``lurkers'' who would
always just consume content silently.

The analysis results also provide insights into how to effectively build and
maintain information communities. The most important takeaway from our
analysis results is that there should be mechanisms that encourage the
formation of a small-core celebrity-follower structure inside the community.
For example, many online social network applications use features such as
``thumb-up'' or ``upvote'' to promote and reward high quality content that
are liked by many community members. Besides providing effective content
filtering (ranking by votes), this voting mechanism also encourages the
optimally stable and efficient community structure: since the production
reward is only earned when a post is upvoted, the members who would produce
low-quality content would not be rewarded and would essentially become the
non-celebrity members in the celebrity-follower structure. The members who
produce high-quality content would be rewarded by the upvotes and becomes
the celebrities in the community. The size of the core of celebrities will
tend to be small if the display of the content in the community is ranked by
popularity: most members will only consume a small portion of the top-ranked
content items therefore only a small set of high quality producers would
actually be rewarded and become the real core of the community. This
analysis would lead to an interesting and counter-intuitive hypothesis: if
the content display of the community is such that different members would
see a diverse range of different items, then this would cause the formation
of a larger-sized celebrity core or a uniform-like structure in the
community which would make the community structure less stable. It would be
interesting to empirically verify if this hypothesis is true in practice.

Another interesting insight is that, in the optimal community structure, the
number of celebrity members in the core, i.e., $\eta\dot n$, must satisfy
that $\eta\cdot n < N_c/N_p$. This means that the size of the of celebrities
core does \emph{not} increase as the size of the community $n$ increases.
This could be a possible reason of why we have communities in the first
place: having a large number of people communicating in a single giant
community is inefficient in terms of the total amount of production because
it only allows a small number of core members to contribute in
content production. Larger total production rate can be achieved by dividing
people into different smaller communities each of which has its own core
members, since the total number of people who will contribute in content
production would be multiplied by the number of communities.

\section{Conclusions}

This paper attempts to obtain a formal understanding of the natural
structural patterns of real-world information communities. We formulate a
mathematical model that describes the generic content production and
consumption behaviours in a community. The analysis result shows that the
small-core celebrity-follower structure is the optimal structure that would
lead to the optimally efficient and stable community. These analytical
results agree with the sociological observations on real-world information
communities. Besides providing a refined microscopic view of the working
mechanisms of information communities, the analysis results also provide
useful insights into how to better build and maintain the structure of
information communities. Designing efficient mechanisms that encourage the
formation of stable and efficient communities would be an interesting topic
for future works.

%
%

\newpage
\noindent\textbf{\Large Appendices}

\appendix

\section{Proof of Theorem~\ref{thm5}}\label{sec-proof-thm5}

In the following two subsections we will prove the two cases separately.

\subsection{Proof of Case 1 and 2}

We will first analyze the strategies for maximizing the social welfare and
then investigate the Nash equilibriums of the such strategies.

\subsection*{Maximum Social Welfare for Case 1 and Case 2}

In this case, $\eta < \frac{N_c}{nN_p} \Rightarrow \eta n N_p < N_c$, i.e.,
the celebrity members are not able satisfy everyone's consumption demand
even if they are producing at full capacity. First, we present the following
lemma.

\begin{lemma}\label{lem4}
The optimal maximum social welfare is reached only if all non-celebrity
	members have $\alpha_i = 0$.
\end{lemma}
\begin{proof}
	This result is easy to see since a product of a non-celebrity member
	never gets consumed by others, a non-celebrity member's optimal strategy
	is always to spend 100\% of their time on consumption.
\end{proof}

\begin{lemma}\label{lem42}
	When the social welfare is maximized, all celebrity members must choose
	the same $\alpha_i$ value.
\end{lemma}
\begin{proof}
	The proof is very similar to part of the proof of Theorem~\ref{thm4} in
	Appendix~\ref{sec-proof-thm4}, therefore is omitted for succinctness.
	The idea is to compare an arbitrary configuration of $\alpha_i$ values
	with the configuration where everyone choose the same value.
\end{proof}

Now let's investigate the $\alpha_i$ value that would maximize the social
welfare for Case 1 and Case 2. There are two cases that are possible when
calculating the social welfare of the community: whether the celebrity
members are under-supplied or over-supplied. For the celebrity members to be
over-supplied, we need the following to hold:
\begin{align}
	& (1-\alpha_i) N_c \le \alpha_i \eta n N_p\\
	\Leftrightarrow\quad & \alpha_i \ge \frac{N_c}{N_c + \eta n N_p}
\end{align}
Otherwise, the celebrity members are under-supplied.

When the celebrity members are over-supplied, the social welfare is
calculated as the following:
\begin{align}\label{thm5c12g1}
	G &= \alpha_i\eta n N_p \left[(1-\eta)n +
	\frac{(1-\alpha_i)N_c}{\alpha_i\eta n N_p}\cdot \eta n\right](r_p + r_c)
	\\
	& = \eta n (r_p + r_c)\left[\alpha_i\left((1-\eta)nN_p - N_c\right) + N_c\right]
\end{align}

When the celebrity members are under-supplied, the social welfare is
calculated as the following:
\begin{align}\label{thm5c12g2}
	G = \alpha_i\eta n^2 N_p
\end{align}

Combining Eq~(\ref{thm5c12g1}) and (\ref{thm5c12g2}), the social welfare for
Case 1 and Case 2 is the following.
\begin{equation}
	G =
	\begin{cases}
		\alpha_i\eta n^2 N_p \quad &\text{if } 0 \le \alpha_i < \frac{N_c}{N_c
	+ \eta n N_p}\\
	\eta n (r_p + r_c)\left[\alpha_i\left((1-\eta)nN_p - N_c\right) +
		N_c\right] \quad &\text{if } \frac{N_c}{N_c + \eta n N_p} \le \alpha_i \le 1
	\end{cases}
\end{equation}
When the celebrity members are under-supplied, $G$ is always an increasing
function of $\alpha_i$. When the celebrity members are over-supplied, $G$
could be either an increase or decreasing function of $\alpha_i$ depending
on whether the coefficient $((1-\eta)nN_p - N_c)$ is positive or negative.
Note that
\begin{equation}
	(1-\eta)nN_p - N_c > 0 \quad\Leftrightarrow\quad \eta < 1 - \frac{N_c}{nN_p}
\end{equation}
Therefore, for Case 1 where $\eta < 1 - \frac{N_c}{nN_p}$, the social
welfare $G$ is maximized when $\alpha_i = 1$, and the maximum social welfare
is
\begin{equation}
	G_{max} = \eta(1-\eta)n^2 N_p (r_p + r_c)
\end{equation}
Therefore, for Case 2 where $\eta \ge 1 - \frac{N_c}{nN_p}$, the social
welfare $G$ is maximized when $\alpha_i = \frac{N_c}{N_c + \eta n N_p}$, and
the maximum social welfare is
\begin{equation}
	G_{\max} = \frac{\eta n^2 N_cN_p(r_p + r_c)}{N_c + \eta n N_p}
\end{equation}

\subsection*{Nash Equilibrium for Case 1}

Now consider the Nash equilibrium for Case 1. For a non-celebrity member,
$\alpha_i = 0$ is clearly the best strategy since there is no additional
reward if they spent any time in production.

For a celebrity member, consider the small deviation from the equilibrium
strategy, i.e., rather than choosing $\alpha_i = 1$, it chooses $\alpha_i =
1 - \delta$ for some $\delta > 0$. This member's product reward is decreased
by $\delta N_p (1-\eta)n r_p$. This member's consumption reward is increased
by up to $\delta N_c r_c$. We want to show that the change of the member's
individual reward to be negative, i.e.,
\begin{align}
	& \Delta R \le \delta N_c r_c - \delta N_p (1-\eta)n r_p < 0\\
	\Leftrightarrow\quad &
	\eta < 1 - \frac{N_cr_c}{nN_pr_p}
\end{align}
which is true according to the assumption of Case 1. Therefore, the Nash
equilibrium of Case 1 is proven.

\subsection*{Nash Equilibrium for Case 2}

Consider a celebrity member who would change their $\alpha_i$ to $\alpha_i
+ \delta$ and $\alpha_i - \delta$ for some $\delta > 0$. In order for a Nash
equilibrium to be formed, it is necessary that both ways of changing result
in a smaller individual reward.

If the strategy is changed to $\alpha_i + \delta$, then the gain of
production reward is up to $\delta N_p n r_p$ and the loss of consumption
reward is $\delta N_c r_c$. In order for the individual reward to decrease,
we need:
\begin{align}
	& \Delta R \le \delta N_p n r_p - \delta N_c r_c < 0\\
	\Leftrightarrow\quad &
	\frac{N_cr_c}{nN_pr_p} > 1\label{thm5c2nash1}
\end{align}

If the strategy is changed to $\alpha_i - \delta$, then the loss of
production reward is $\delta N_p n r_p$ and the gain of consumption reward
is up to $\delta N_c r_c$. In order for the individual reward to decrease,
we need:
\begin{align}
	& \Delta R \le \delta N_c r_c - \delta N_p n r_p < 0\\
	\Leftrightarrow\quad &
	\frac{N_cr_c}{nN_pr_p} < 1\label{thm5c2nash2}
\end{align}

A Nash equilibrium requires that both (\ref{thm5c2nash1}) and
(\ref{thm5c2nash2}) to be true, which is impossible, therefore the Nash
equilibrium cannot exist. The Nash equilibrium of Case 1 is proven, which
completes all proofs for Case 1 and Case 2.

\subsection{Proof of Case 3 and Case 4}

Case 3 and Case 4 are the cases where the celebrity members are capable of
supplying any member in community, i.e., $\eta n N_p \ge N_c$. First of
all, Lemma~\ref{lem4} and Lemma~\ref{lem42} both apply to this part of the
proof. And we have the following additional lemma.

\begin{lemma}\label{lem3}
	When the social welfare is maximized, a non-celebrity member must not be
	over-supplied, and a celebrity member must not be under-supplied.
\end{lemma}
\begin{proof}
	Suppose the non-celebrity members were over-supplied, then the
	high-quality members must also be over-supplied. Reducing production
	time will increase the total number of consumptions for sure, therefore
	the global welfare would be increased.

	Similarly, suppose the celebrity members were under-supplied, then the
	non-celebrity members must also be under-supplied. Increasing production
	time will increase the total number of consumptions for sure, therefore
	the global welfare would be increased.
\end{proof}

\subsection*{Maximum Social Welfare for Case 3}

We now write down the expression of the social welfare. Because of
Lemma~\ref{lem42}, we let $\alpha_h$ be the common $\alpha_i$ value of all
celebrity members, then the total number of produced items is $\alpha_h \eta
n N_p$. For each item, the expected number of consumptions from the
non-celebrity members is $(1-\eta)n$ since each non-celebrity member is not
over-supplied; the expected number of consumptions from the celebrity
members is the following (because they are not under-supplied):
\begin{equation}
	\frac{(1-\alpha_h)N_c}{\alpha_h N_p \eta n}\cdot \eta n
	= \frac{N_c}{N_p}\left(\frac{1}{\alpha_h}-1\right)
\end{equation}
Therefore, total reward is the $(r_p + r_c)$ multiplied by the total
number of consumptions, which is
\begin{align}
	R &= (r_p+r_c)\cdot \alpha_hN_p \eta n\cdot \left[(1-\eta)n +
	\frac{N_c}{N_p}\left(\frac{1}{\alpha_h}-1\right)\right] \\
	&= (r_p + r_c)\cdot \left[\alpha_h\left[N_p \eta(1-\eta)n^2) - N_c \eta
	n\right] + N_c \eta n\right] \label{pf5c1e1}
\end{align}
The expression in (\ref{pf5c1e1}) could be an increasing or decreasing
function of $\alpha_h$. It depends the sign of coefficient of $\alpha_h$.
Therefore, we need to divide into two sub-cases, which are exactly Case 3
and Case 4.

Case 3: The coefficient is non-negative, i.e.,
\begin{align}
	&\left[N_p \eta(1-\eta)n^2 - N_c \eta n\right] \ge 0 \\
	\Leftrightarrow\quad
	&\eta \le 1 - \frac{N_c}{N_p}
\end{align}
In this case, to maximize the global welfare, $\alpha_h$ should be as large
as possible, i.e., $\alpha_h$ should be the largest possible value that
keeps the non-celebrity members non-oversupplied. It is the value that
generates exactly $N_c$ items, which is
\begin{equation}
	\alpha_h = \frac{N_c}{\eta n N_p}
\end{equation}
Plug the above value into Eq.~(\ref{pf5c1e1}), the maximum social welfare is
\begin{align}
	G_{\max} = (r_p + r_c)N_c\left(n - \frac{N_c}{N_p}\right)
\end{align}

\subsection*{Nash Equilibrium for Case 3}

Now investigate the Nash equilibrium for Case 3 by consider the following
two types of changes (let $\delta > 0$ denote the change amount):a celebrity
member increasing from $\alpha_i$ to $\alpha_h + \delta$; and a celebrity
member decreasing from $\alpha_h$ to $\alpha_h - \delta$.

If the strategy is changed to $\alpha_h + \delta$, then the gain of
production reward is $\delta N_p ((1-\eta n)n + \eta n (1-\alpha_h))r_p$ and
the loss of consumption reward is $\delta N_c r_c$. In order for the
individual reward to decrease, we need:
\begin{align}
	& \Delta R = \delta N_p ((1-\eta n)n + \eta n (1-\alpha_h))r_p - \delta
	N_c r_c < 0\\
	\Leftrightarrow\quad &
	\frac{N_cr_c}{(nN_p-N_c)r_p} > 1\label{thm5c3nash1}
\end{align}

If the strategy is changed to $\alpha_h + \delta$, then the loss of
production reward is $\delta N_p ((1-\eta n)n + \eta n (1-\alpha_h))r_p$ and
the gain of consumption reward is $\delta N_c r_c$. In order for the
individual reward to decrease, we need:
\begin{align}
	& \Delta R = \delta N_c r_c\delta N_p ((1-\eta n)n + \eta n
	(1-\alpha_h))r_p  < 0\\
	\Leftrightarrow\quad &
	\frac{N_cr_c}{(nN_p-N_c)r_p} < 1\label{thm5c3nash2}
\end{align}

A Nash equilibrium requires that both (\ref{thm5c3nash1}) and
(\ref{thm5c3nash2}) to be true, which is impossible, therefore the Nash
equilibrium cannot exist. The Nash equilibrium of Case 3 is proven.

\subsection*{Maximum Social Welfare for Case 4}

Case 4: The coefficient in (\ref{pf5c1e1}) is negative, i.e.,
\begin{align}
	&\left[N_p\eta(1-\eta)n^2 - N_c\eta n\right] < 0 \\
	\Leftrightarrow\quad
	&\eta > 1 - \frac{N_c}{nN_p}
\end{align}
In this case, to maximize the social welfare, $\alpha_h$ should be as small
as possible, i.e., $\alpha_h$ should be the smallest possible value that
keeps the celebrity members non-under-supplied. The number of produced
items should be exactly $(1-\alpha_h)N_c$ items, i.e.,
\begin{align}
	&\alpha_h N_p\eta n = (1-\alpha_h)N_c \\
	\Leftrightarrow\quad
	&\alpha_h = \frac{N_c}{N_c + \eta n N_p}
\end{align}
Plug the above value into Eq.~(\ref{pf5c1e1}), the maximum global welfare is
\begin{equation}
	G_{\max} = \frac{\eta n^2 N_cN_p(r_p + r_c)}{N_c + \eta n N_p}
\end{equation}

Note that, Case 4's optimal state for social welfare is in fact exactly the
same as that of Case 2. Therefore, the Nash equilibrium of Case 4 follows
the same configuration as that of Case 2, which means this set of strategy
cannot form a Nash equilibrium. Hence, we have completed the proof for
Theorem~\ref{thm5}.

\section{Proof of Theorem~\ref{thm4}}\label{sec-proof-thm4}

The claim is that the maximum social welfare is reached when (1) all members
have the same $\alpha_i$ value and (2) the number of produced items is just
enough every member's consumption. We consider three cases when calculating
the social welfare of the community: everyone is over-supplied, everyone is
under-supplied and some people are over-supplied while others are
under-supplied. Let $s = \sum_{i=1}^n \alpha_i$ denote the sum of everyone's
$\alpha_i$ value.

Case 1: Everyone is over-supplied with content.
The consumption reward a member $i$ is simply $(1-\alpha_i)N_cr_c$, therefore
the global total consumption reward is
\begin{equation}\label{pf4c1e1}
	\sum_{i=1}^n (1-\alpha_i)N_cr_c = N_cr_c(n-\sum_{i=1}^n \alpha_i) =
	N_cr_c(n-s)
\end{equation}
The global total production reward is equal to the total number of
consumption multiplied by $r_p$, which is
\begin{equation}\label{pf4c1e2}
	r_p\cdot \sum_{i=1}^n (1-\alpha_i)N_c = N_cr_p\left(n-\sum_{i=1}^n
	\alpha_i\right) = N_cr_p(n-s)
\end{equation}
Therefore the total global welfare for Case 1 is the sum of (\ref{pf4c1e1})
and (\ref{pf4c1e2}) which is
\begin{equation}
	G(\vec{\alpha}) = N_cr_c(n-s) + N_cr_p(n-s) = N_c(r_p +
	r_c)\left(n-\sum_{i=1}^n \alpha_i\right)
\end{equation}
We can see that, in the case of everyone being over-supplied, the social
welfare would decrease if anybody increases their $\alpha_i$; the optimal
value for this case is reached when everyone reduces their $\alpha_i$ as
much as possible as long as everyone is still over-supplied.

Now we want to show that the global welfare for this case is no greater than
the maximum global welfare shown in in Equation~(\ref{thm1globmax}) of
Theorem~\ref{thm4}, i.e., we want to show the following
\begin{align}
	& N_c(r_p + r_c)\left(n-\sum_{i=1}^n \alpha_i\right) \le
	\frac{\eta n^2N_pN_c(r_c+r_p)}{N_c + \eta nN_p} \\
	\Leftrightarrow \qquad &
	\left(n-\sum_{i=1}^n \alpha_i\right) \le
	\frac{\eta n^2N_p}{N_c + \eta nN_p} = n\cdot \left(1-\frac{N_c}{N_c+\eta
	nN_p}\right) \\
	\Leftrightarrow \qquad &
	\left(\sum_{i=1}^n \alpha_i\right) \ge
	n\cdot \frac{N_c}{N_c+\eta nN_p}\label{pf4c1e3}
\end{align}
Note that $N_c/(N_c+\eta nN_p)$ is the optimal $\alpha_i$ value which make
the number of produced items exactly the same as the number needed by every
member. The inequality in (\ref{pf4c1e3}) must be true because if it were
not true, it would cause a member to be under-supplied, which violates the
assumption for this case. Completing the proof for Case 1.

Case 2: Everyone is under-supplied with content. In this case, each member
consumes all content items that are available for them, therefore the total
consumption reward is
\begin{equation}
	\eta nN_pr_c\sum_{i=1}^n \alpha_{i} = \eta nN_pr_c s
\end{equation}

The total production reward is again the total number of consumptions
multiplied by $r_p$. Each of the produced items is consumed $\eta n$ times,
thus we have the total production reward as follows:
\begin{equation}
	\eta nr_p\sum_{i=1}^n \alpha_i N_p = \eta N_pr_pns
\end{equation}
Summing up the total consumption reward and the total production rewards,
the total global welfare for Case 2 is
\begin{equation}
	G(\vec{\alpha}) = \eta N_p(r_p+r_c)n s = \eta
	N_p(r_p+r_c)n\left(\sum_{i=1}^n \alpha_i\right)
\end{equation}
We can see that, in the case of everyone being under-supplied, the social
welfare would decrease if anybody decreases their $\alpha_i$; the optimal
value for this case is reached when everyone increase their $\alpha_i$ as
much as possible as long as everyone is still under-supplied.

Again, we want to show that the global welfare for this case is no greater
than the maximum global welfare shown in in Equation~(\ref{thm1globmax}) of
Theorem~\ref{thm4}, i.e., we want to show the following
\begin{align}
	& \eta N_p(r_p+r_c)n\left(\sum_{i=1}^n \alpha_i\right) \le
	\frac{\eta n^2N_pN_c(r_c+r_p)}{N_c + \eta nN_p} \\
	\Leftrightarrow \qquad &
	\sum_{i=1}^n \alpha_i \le n\cdot \frac{N_c}{N_c + \eta nN_p} \label{pf4c2e1}
\end{align}
Again, since $N_c/(N_c+\eta nN_p)$ is the $\alpha_i$ value that makes the
number of produced items exactly the same as the number needed by every
member. The inequality in (\ref{pf4c2e1}) must be true because if it were
not true, it would cause a member to be over-supplied, which violates the
assumption for this case. Completing the proof for Case 2.

Case 3 is in between Case 1 and Case 2, i.e., we assume that a portion of
the members in the community are over-supplied while others are
under-supplied.

Let $\bar{\alpha} = \sum_{i=1}^n \alpha_i / n$ be the average value of all
members $\alpha_i$ values. We will compare an arbitrary configuration in
Case 3 with the one where everyone choose the same $\alpha_i =
\bar{\alpha}$. Note that, if everyone choose the same $\alpha_i$, it is
either Case 1 or Case 2, i.e., either everyone is over-supplied or everyone
is under-supplied. We want to show that any configuration in Case 3 has
smaller social welfare than the configuration where everyone has the same
$\alpha_i = \bar{\alpha}$. We call this two configurations the ``original
configuration'' and the ``average configuration''.

Because of the choice of the value of $\bar{\alpha}$, the sum of everyone's
$\alpha_i$ value stay the same, therefore the total number of produced items
is the same between the two configurations. We will compare the social
welfares of the two configurations in two cases.

Case 3.1: Everyone is over-supplied in the average configuration. Then
everyone consumes exactly $N_c(1-\bar{\alpha})$ items in the average
configuration. Therefore, the consumption reward for each member is
$N_cr_c(1-\bar{\alpha})$. In the original configuration, some members have
greater $\alpha_i$ values than $\bar{\alpha}$ and some others have smaller
$\alpha_i$ values than $\bar{\alpha}$. Consider a pair of members of which
member A has $\alpha_i = \bar{\alpha} + \delta$ and member B has $\alpha_i =
\bar{\alpha} - \delta$. For member A, since she becomes even more
over-supplied and spends less time on consumption, her consumption reward
decreases by exactly $\delta N_c r_c$. On the other hand, member B's
consumption reward increases by \emph{at most} $\delta N_c r_c$, depending
on whether B becomes under-supplied after decreasing her $\alpha_i$ value
from $\bar{\alpha}$. Therefore, the all-members total consumption reward of
the original configuration is no greater that of the average configuration.

In terms of the total production reward, recall that it is the total number
of consumptions multiplied by $r_p$. Let $K$ denote the total number of
items produced. In the average configuration, for each item, the expected
number of consumptions contributed by each member is
$(1-\bar{\alpha}N_c)/K$. In the original configuration, consider again
member A with $\alpha_i = \bar{\alpha}+\delta$ and member B with $\alpha_i =
\bar{\alpha}-\delta$. Member A's contribution decreases by exactly $\delta
N_c/K$, while member B's contribution increases by \emph{at most} $\delta
N_c/K$. Therefore, the all-members total production reward of the original
configuration is no greater that of the average configuration. Thus, we can
conclude for Case 3.1 that the global welfare of the original configuration
is no greater than that of the average configuration.

Case 3.2: Everyone is under-supplied in the average configuration. Then
everyone consumes exactly $K = \eta N_p\sum_{i=1}^n \alpha_i$ items in the
average configuration. Therefore, the total consumption reward is $nKr_c$.
Now consider the original configuration with member A having $\alpha_i =
\bar{\alpha} + \delta$ and member B having $\alpha_i = \bar{\alpha}
-\delta$. For member B, she is still under-supplied and her consumption
reward stays the same since the total number of items stays the same. For
member A, her consumption reward stays the same of she is still
under-supplied and it decreases if she becomes over-supplied and consumes
less than $K$ items. Therefore, the all-members total consumption reward of
the original configuration is no greater that of the average configuration.

In terms of the total production reward, the average configuration has a
total production reward of $nKr_p$, whereas the original configuration can
have the same total production rewards, or less if anyone becomes
over-supplied by having a larger $\alpha_i$ and consumes less $K$ items.
Therefore, the all-members total production reward of the original
configuration is no greater that of the average configuration. Thus, we can
conclude for Case 3.2 that the social welfare of the original configuration
is no greater than that of the average configuration.

Since the average configuration belongs to either Case 1 or Case 2, both of
which are proven to have no greater global welfare than the maximum we
claimed in (\ref{thm1globmax}). Hence, the global welfare in
(\ref{thm1globmax}) is the maximum possible global welfare of the community.

\subsection*{Nash Equilibrium for the Uniform Structure}

When every member chooses $\alpha_i = N_c/(N_c + \eta nN_p)$, the
consumption reward of a member is
\begin{equation}
	R_c\left(\alpha_i = \frac{N_c}{N_c + \eta nN_p}\right) =
	(1-\alpha_i)N_cr_c = \frac{\eta nN_pN_cr_c}{N_c + \eta nN_p}
\end{equation}
The production reward of a member is
\begin{equation}
	R_p\left(\alpha_i = \frac{N_c}{N_c + nN_p}\right) = \alpha_iN_pr_p\eta n =
	\frac{\eta nN_pN_cr_p}{N_c + \eta nN_p}
\end{equation}
The total reward is the sum of the above two which is
\begin{equation}
	R\left(\alpha_i = \frac{N_c}{N_c + \eta nN_p}\right) =
	\frac{\eta nN_pN_c(r_c+r_p)}{N_c + \eta nN_p}
\end{equation}

Now we calculate the reward for member $i$ when she changes her strategy is
changed to $\alpha_i + \delta$. We divide the analysis into two cases:
$\delta > 0$ and $\delta < 0$.

Case 1: $\delta > 0$, i.e., member $i$ increases $\alpha_i$. This will make
some members over-supplied. For member $i$, the consumption reward decreases
because of the reduced amount of consumption time, i.e.,
\begin{equation}
	\Delta R_c = -\delta N_c r_c
\end{equation}
The total number of available items for the followers is now $\eta
\alpha_i n N_p + \delta N_p$, the probability of an item being chosen is
\begin{equation}
	\frac{\eta \alpha_i n N_p}{\eta \alpha_i n N_p + \delta N_p}
	= \frac{\eta \alpha_i n}{\eta \alpha_i n + \delta}
\end{equation}
The production reward for member $i$ becomes
\begin{equation}
	(\alpha_i + \delta)\eta n\cdot \frac{\eta \alpha_i n}{\eta \alpha_i n +
	\delta}\cdot N_pr_p
\end{equation}
Therefore, the change in member $i$'s total reward is
\begin{equation}
	\Delta R =
	(\alpha_i + \delta)\eta n\cdot \frac{\eta \alpha_i n}{\eta \alpha_i n +
	\delta}\cdot N_pr_p - \eta \alpha_i nN_pr_p - \delta N_c r_c
\end{equation}
We want to check the condition for $\Delta R < 0$. In the following
derivation, we let $K = N_pr_p / N_cr_c$.
\begin{align}
	&\Delta R =
	(\alpha_i + \delta)\eta n\cdot \frac{\eta \alpha_i n}{\eta \alpha_i n +
	\delta}\cdot N_pr_p - \eta \alpha_i nN_pr_p - \delta N_c r_c < 0 \\
	\Rightarrow \quad &
	K\left(\frac{\alpha_i\eta^2 n^2(\alpha_i +
	\delta)}{\alpha_i\eta n+\delta}-\alpha_i\eta n\right) < \delta \\
	\Rightarrow \quad &
	\frac{K\alpha_i \eta n\delta(\eta n-1)}{\alpha_i\eta n+\delta} < \delta \\
	\Rightarrow \quad &
	\frac{K\alpha_i\eta n(\eta n-1)}{\alpha_i\eta n+\delta} < 1 \qquad \#\;
	\delta > 0 \\
	\Rightarrow \quad &
	\delta > \alpha_i \eta n (K\eta n-K-1)\label{deltacond1}
\end{align}
To make sure (\ref{deltacond1}) is satisfied, we must have
\begin{align}
	& \alpha_i \eta n (K\eta n-K-1) \le 0 \\
	\Rightarrow \quad &
	(K\eta n-K-1) \le 0 \\
	\Rightarrow \quad &
	K\le \frac{1}{\eta n-1} \\
	\Rightarrow \quad &
	\frac{N_pr_p}{N_cr_c} \le \frac{1}{\eta n-1}\\
	\Rightarrow \quad &
	\eta \le \left(\frac{N_cr_c}{nN_pr_p} + \frac{1}{n}\right)\label{pf3c1cond}
\end{align}
Summary of Case 1: If (\ref{pf3c1cond}) is true, then the total reward for member
$i$ decreases if she increases her $\alpha_i$ to $\alpha_i+\delta\; (\delta
> 0)$.

Case 2: $\delta < 0$, i.e., member $i$ decreases her $\alpha_i$, then
everyone in the community become under-supplied. For member $i$, the
production reward is reduced simply because of the reduced value of
$\alpha_i$, i.e.,
\begin{equation}
	\Delta R_p = \delta \eta n N_pr_p
\end{equation}
The consumption reward for member $i$ is also reduced because of the reduced
number of items available for consumption (even though member $i$ tries to
spend more time on consumption).
\begin{equation}
	\Delta R_c = \delta N_p r_c
\end{equation}
Therefore the change in the total reward for member $i$ is
\begin{equation}
	\Delta R = \delta\cdot N_p(\eta nr_p + r_c) < 0 \qquad \#\; \delta < 0
\end{equation}
Case 2 done.

Combining Case 1 and Case 2, if and only if (\ref{pf3c1cond}) is true, the
change of $\alpha_i$ will result in a decrease of member $i$'s total reward,
therefore the strategy of choosing the original $\alpha_i$ results in a
Nash equilibrium. Completing the proof for Theorem~\ref{thm4}.

\section{Proof of Theorem~\ref{thm6}}\label{sec-proof-thm6}

\begin{proof}
	The $G_{\max\text{-celebrity}}$ in Eq~(\ref{gmax-celeb}) is greater than
	or equal to the $G_{\max\text{-uniform}}$ in Eq~(\ref{thm1globmax}) if
	and only if the following is true.
	\begin{align}
		& \eta(1-\eta)n^2 N_p (r_p + r_c) \ge
		\frac{\eta n^2N_pN_c(r_c+r_p)}{N_c + \eta nN_p}\\
		\Leftrightarrow\quad &
		1-\eta \ge \frac{N_c}{N_c + \eta n N_p}\\
		\Leftrightarrow\quad &
		(1-\eta)(N_c + \eta n N_p) \ge N_c\\
		\Leftrightarrow\quad &
		\eta (nN_p - N_c + n\eta N_p) \ge 0\\
		\Leftrightarrow\quad &
		nN_p - N_c + n\eta N_p \ge 0 \\
		\Leftrightarrow\quad &
		\eta \ge \frac{N_c - nN_p}{nN_p}\label{pf6cond}
	\end{align}
	According to the modelling assumption in Eq~(\ref{assum-nc}), $N_c -
	nN_p \le 0$, and since $\eta \ge 0$, the condition is (\ref{pf6cond}) is
	always true. Completing the proof of Theorem~\ref{thm6}.
\end{proof}

\end{document}